\documentclass{amsart}

\usepackage{xcolor}
\definecolor{lgray}{RGB}{240,240,240}
\definecolor{webgreen}{rgb}{0,.5,0}
\definecolor{webbrown}{rgb}{.6,0,0}
\definecolor{RoyalBlue}{cmyk}{1, 0.50, 0, 0}
\usepackage[colorlinks=true, breaklinks=true, urlcolor=webbrown, linkcolor=RoyalBlue, citecolor=webgreen, backref=page]{hyperref}

\usepackage{subcaption, graphicx}
\usepackage{verbatim, setspace}
\usepackage{amsmath, amssymb, bm}

\usepackage{newtxtext,newtxmath}

\newtheorem{theorem}{Theorem}[section]
\newtheorem{corollary}[theorem]{Corollary}

\newcommand{\C}		{\mathbb{C}}
\newcommand{\N}		{\mathbb{N}}

\newcommand{\dist}{\mathrm{dist}}

\newcommand{\sgn}{\mathrm{sgn}}
\renewcommand{\arg}{\mathrm{arg}}
\renewcommand{\det}{\mathrm{det}}

\newcommand{\Ai}{\mathrm{Ai}}
\newcommand{\Bi}{\mathrm{Bi}}

\newcommand{\qandq}{\quad \text{and} \quad}

\newcommand{\ic}{\mathrm{i}}

\begin{document}

\title{On Airy Solutions of P$_\mathrm{II}$ and Complex Cubic Ensemble of Random Matrices, I}

\author{Ahmad Barhoumi}
\address{Department of Mathematics, Royal Institute of Technology (KTH), Stockholm, Sweden}
\email{ahmadba@kth.se}

\author{Pavel  Bleher}
\address{Department of Mathematical Sciences, Indiana University Indianapolis, 402~North Blackford Street, Indianapolis, IN 46202, USA}
\email{pbleher@iupui.edu}

\author{Alfredo Dea\~no}
\address{Department of Mathematics, Universidad Carlos III de Madrid, Avda. de la Universidad 30, 28911 Legan\'es, Madrid, Spain}
\email{alfredo.deanho@uc3m.es}

\author{Maxim Yattselev}
\address{Department of Mathematical Sciences, Indiana University Indianapolis, 402~North Blackford Street, Indianapolis, IN 46202, USA}
\email{maxyatts@iu.edu}

\thanks{The third author acknowledges financial support from Direcci\'on General de Investigaci\'on e Innovaci\'on, Consejer\'ia de Educaci\'on e Investigaci\'on of Comunidad de Madrid (Spain), and Universidad de Alcal\'a under grant CM/JIN/2021-014, as well as from grant PID2021-123969NB-I00, funded by MCIN/AEI/10.13039/501100011033, and from grant PID2021-122154NB-I00 from Spanish Agencia Estatal de Investigaci\'on. The research of the last author was supported in part by a grant from the Simons Foundation, CGM-706591.}



\begin{abstract}
We show that the one-parameter family of special solutions of  P$_\mathrm{II}$, the second Painlev\'e equation, constructed from the Airy functions, as well as associated solutions of P$_\mathrm{XXXIV}$ and S$_\mathrm{II}$, can be expressed via the recurrence coefficients of orthogonal polynomials that appear in the analysis of the Hermitian random matrix ensemble with a  cubic potential. Exploiting this connection we show that solutions of  P$_\mathrm{II}$ that depend only on the first Airy function \( \Ai \) (but not on \( \Bi \)) possess a scaling limit in the pole free region, which includes a disk around the origin whose radius grows with the parameter. We then use the scaling limit to show that these solutions are monotone in the parameter on the negative real axis.
\end{abstract}

\maketitle

\section{Introduction}

At the beginning of the 20th century, it was shown by Painlev\'e and Gambier that among the second order differential equations of the form
\[
q^{\prime\prime}= F(q^\prime,q,z),
\]
where \( F \) is rational in \( q,q^\prime \) and analytic in \( z \), there are exactly $50$ canonical equations whose solutions do not possess movable branch points. Among these, $44$ can either be reduced to linear equations, equations solved in terms of elliptic functions, or to the remaining $6$ equations. In this work we are interested in what is now known as the second Painlev\'e equation, see \eqref{P2}. It was shown by Gambier \cite{MR1555055} that P$_\mathrm{II}$ has a one-parameter family of solutions that are expressible in terms of the Airy functions and their derivatives. These special function solutions appear in many contexts. For example, in Hermitian random matrix ensembles as the asymptotic analysis of eigenvalue behavior (in a suitable double scaling limit as the size of the matrices tends to infinity and near the edge of the spectrum) makes use of a correlation kernel that in particular cases can be written in terms of Airy solutions of P$_\mathrm{II}$ and P$_\mathrm{XXXIV}$, see \cite{MR2429248}. Another appearance of these functions in random matrix theory comes from Hermitian ensembles with cubic potential. The cubic random matrix model, given by the probability distribution
\[
\tilde Z_N^{-1}(u) e^{-N\mathrm{Tr}(M^2/2-uM^3)}dM
\]
on \( N\times N \) Hermitian matrices \( M \), where \( u \) is a parameter, has been investigated in physical literature by Br\'ezin, Itzykson, Parisi, and Zuber \cite{BIPZ} and Bessis, Itzykson, and Zuber \cite{BIZ}. There it was recognized that the free energy of this model must possess an expansion in the parameter \( N \). The \( N^{-2g} \) coefficient of this expansion, as a function of \( u \), possesses a Taylor series around the origin, whose coefficients count the number of the three-valent  graphs on a Riemann surface of genus \( g \). These considerations were later made rigorous by the middle two authors of the present work \cite{MR3071662}. Further investigations of this model were carried out by Huybrechs, Kuijlaars, and Lejon in \cite{MR3218792,MR4031470} and by the authors in \cite{MR3493550,MR3607591,MR4436195}, see also \cite{MR3090748,MR3319490} for the numerical studies. The main tool of investigation in these works was the connection to non-Hermitian orthogonal polynomials with respect to the cubic weight \( \exp(-z^3/3+tz) \) (parameter \( t \) is connected to the parameter \( u \)), whose moments can be expressed via the Airy functions \( \Ai(z) \) and \( \Bi(z) \) as well as their derivatives. In this work we utilize this connection to show what the results of \cite{MR3607591} can say about the solutions of  P$_\mathrm{II}$ that depend only on \( \Ai(z) \) and its derivatives. The techniques of \cite{MR3071662,MR3493550,MR3607591,MR4436195} can be extended to study all special solutions and we shall do so in a subsequent publication.

\section{Airy Solutions of P$_\mathrm{II}$}

Each of the Painlev\'e equations P$_\mathrm{I}$--P$_\mathrm{VI}$ can be written as a Hamiltonian system
\[
\frac{dq}{dz} = \frac{\partial H}{\partial p} \qandq \frac{dp}{dz} = -\frac{\partial H}{\partial q},
\]
see \cite{MR596006,MR581468,MR625446}. In the case of P$_\mathrm{II}$, the Hamiltonian \( H_\mathrm{II} \) is equal to
\[
H_\mathrm{II}(q,p,z;\alpha) = \frac12p^2 - \left(q^2 + \frac z2\right) p - \left(\alpha+\frac12\right)q,
\]
where \( \alpha \) is a parameter. Hence, the Hamiltonian system becomes
\begin{equation}
\label{ham_sys}
q^\prime = p-q^2 - z/2 \qandq p^\prime = 2pq + \alpha + 1/2.
\end{equation}
Eliminating \( p \) from \eqref{ham_sys} gives P$_\mathrm{II}$:
\begin{equation}
\label{P2}
q^{\prime\prime} = 2q^3 + zq + \alpha,
\end{equation}
while eliminating \( q \) from \eqref{ham_sys} yields P$_\mathrm{XXXIV}$:
\begin{equation}
\label{P34}
p^{\prime\prime} = \frac{( p^\prime )^2 - (\alpha+1/2)^2}{2p} + 2p^2 - zp.
\end{equation}
Notice that the first equation \eqref{ham_sys} allows one to express \( p \) through \( q,q^\prime \) and \( z \) while the second equation allows to express \( q \) through \( p,p^\prime \). These connection formulae provide one-to-one correspondence between solutions of P$_\mathrm{II}$ and P$_\mathrm{XXXIV}$. Morever, if \( p_{\alpha+1/2}(z),q_{\alpha+1/2}(z) \) are solutions of \eqref{ham_sys} and
\begin{equation}
\label{sigma1}
\sigma_{\alpha+1/2}(z) := H_\mathrm{II} \big( q_{\alpha+1/2}(z),p_{\alpha+1/2}(z),z;\alpha \big),
\end{equation}
then this function solves S$_\mathrm{II}$, the Jimbo-Miwa-Okamoto \( \sigma \)-form of P$_\mathrm{II}$:
\begin{equation}
\label{S2}
( \sigma^{\prime\prime})^2 + 4( \sigma^\prime )^3 + 2 \sigma^\prime ( z \sigma^\prime -\sigma ) = (\alpha/2+1/4)^2.
\end{equation}
Conversely, if \( \sigma_{\alpha+1/2}(z) \) solves \eqref{S2}, then the functions
\begin{equation}
\label{sigma2}
q_{\alpha+1/2}(z) = \frac{2\sigma_{\alpha+1/2}^{\prime\prime}(z)+\alpha+1/2}{4\sigma_{\alpha+1/2}^\prime(z)} \qandq p_{\alpha+1/2}(z) = -2\sigma_{\alpha+1/2}^\prime(z), 
\end{equation}
solve \eqref{P2} and \eqref{P34}, respectively, see \cite{MR596006,MR581468,MR625446,MR854008}.

It is known from the work of Gambier, see \cite{MR1555055}, that \eqref{P2}, and therefore \eqref{P34} and \eqref{S2} via the correspondence \eqref{ham_sys}, \eqref{sigma1}, and \eqref{sigma2}, has a one-parameter family of solutions expressible through the Airy functions and their derivatives if and only if \( \alpha +1/2 \) is an integer. Recall that the standard Airy functions can be given by their integral representations
\begin{equation}
\label{Ai}
\Ai(z) = \frac1{2\pi\ic} \int_{L_1-L_2} e^{-V(s;z)} d s, \quad V(s;z) := -\frac{s^3}3 + sz,
\end{equation}
where \( L_k := \big\{ xe^{\pi\ic(-1+2k/3)}:x\in(0,\infty)\big\} \), \( k\in\{0,1,2\} \), are rays oriented towards the origin, and
\begin{equation}
\label{Bi}
\Bi(z) = \frac1{2\pi} \left(\int_{L_0-L_1} + \int_{L_0-L_2}\right) e^{-V(s;z)} d s.
\end{equation}
When \( \alpha = 1/2 \), it can be shown that \eqref{P2} has a one-parameter family of solutions given by
\[
q_1(z;\lambda) = - \frac{d}{dz} \log \left( C_1 \, \Ai\left(-2^{-1/3}z\right) + C_2 \, \Bi\left(-2^{-1/3}z\right) \right),
\]
where \( \lambda\in\overline\C \) and \( C_1,C_2\in\C \) are any numbers such that \( C_2/C_1 =\lambda\in\C \) and \( C_1=0 \) when \( \lambda=\infty \) (the branch of the logarithm is not important as we immediately take \( z \) derivative). The solutions \( q_n(z;\lambda) \) corresponding to \( \alpha=n-1/2 \), \( n\geq 1 \), are constructed recursively via B\"acklund transformation
\[
q_{n+1}(z;\lambda) = -q_n(z;\lambda) - \frac{2n}{2q_n^2(z;\lambda)+2q_n^\prime(z;\lambda)^2+z}.
\]
The solutions for non-positive values of \( n \) are then obtained via \( q_n(z;\lambda) = -q_{-n+1}(z;\lambda) \). However, the above formula is not the most convenient way of expressing these solutions. To this end, define \( \tau_0(z;\lambda)\equiv1 \) and
\begin{equation}
\label{taun}
\tau_n(z;\lambda) := \det \left[ \frac{d^{j+k}}{dz^{j+k}} \left( C_1 \, \Ai\left(-2^{-1/3}z\right) + C_2 \, \Bi\left(-2^{-1/3}z\right) \right) \right]_{j,k=0}^{n-1}
\end{equation}
when \( n\in\N \). Then, see \cite{MR588248,MR854008,MR3529955}, the solutions \( q_n(z;\lambda) \), \( n\geq 1 \), and the corresponding functions \( p_n(z;\lambda) \) and \( \sigma_n(z;\lambda) \) solving \eqref{P34} and \eqref{S2}, respectively, can be expressed as
\begin{equation}
\label{Airy_sol}
\begin{cases}
q_n(z;\lambda) & \displaystyle = \frac{d}{dz}\log\frac{\tau_{n-1}(z;\lambda)}{\tau_n(z;\lambda)}, \smallskip \\
p_n(z;\lambda) & \displaystyle = -2\frac{d^2}{dz^2}\log\tau_n(z;\lambda), \smallskip \\
\sigma_n(z;\lambda) & \displaystyle = \frac{d}{dz} \log \tau_n(z;\lambda).
\end{cases}
\end{equation}

It was pointed out by Clarkson in \cite{MR3529955}, see also \cite{MR3260257}, based on the numerical computations, that the functions \( q_n(z;0) \) seemed to be tronqu\'ee solutions of P$_\mathrm{II}$; that is, they have no poles in a sector of the complex plane.  Numerical computations of Fornberg and Weideman in \cite{MR3260257} also suggested that for general \( \lambda \) there are three sectors in which the solutions \(  q_n(z;\lambda) \) are pole free, see Figure~\ref{fig:zp}. These conjectures were proven for \( z \) large by the third author in \cite{MR3860606}, where large \( z \) expansions in the pole free sectors were obtained. In this note we further provide scaling limits of the purely Airy solutions. We shall obtain scaling limits of all special solutions in a subsequent publication. Moreover, it was suggested by Clarkson in \cite{MR3529955}, see also \cite{MR3516120}, again, based on numerical computations, that \( q_{n+1}(z;0)<q_n(z;0) \) for \( z<0 \). We confirm this conjecture for all \( n \) large.

\begin{figure}[ht!]
\centering
\begin{subfigure}{.45\textwidth}
\includegraphics[width=\textwidth]{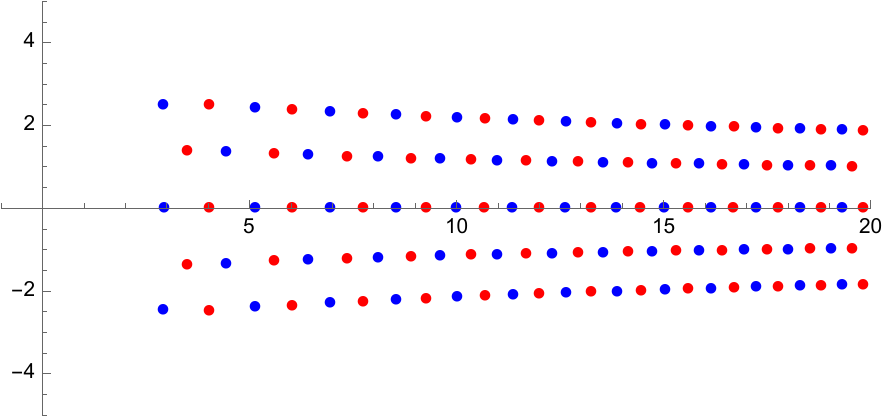}
\end{subfigure}
\quad
\begin{subfigure}{.35\textwidth}
\includegraphics[width=\textwidth]{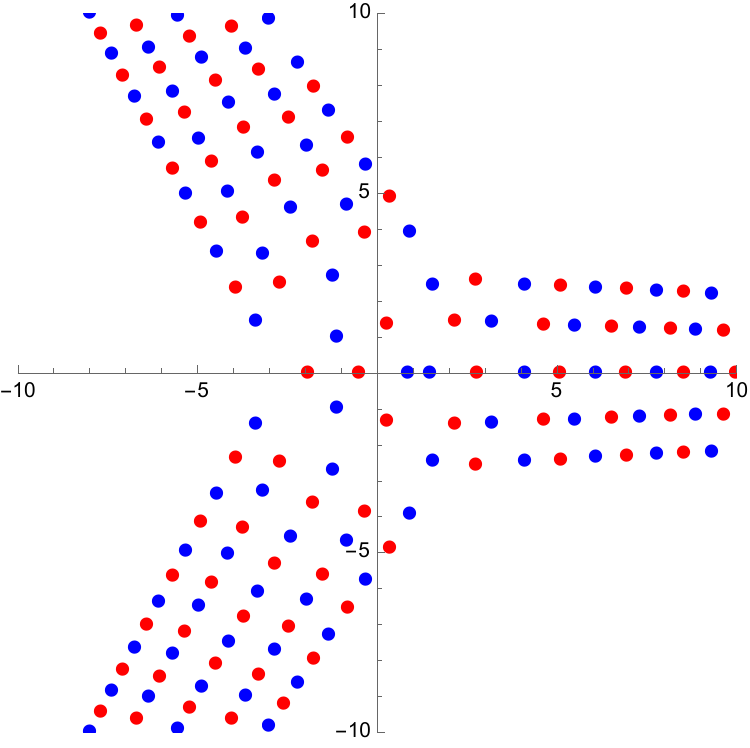}
\end{subfigure}
\caption{\small Zeros (blue) and poles (red) of \( q_3(z;0) \) (left panel) and \( q_3(z;\infty) \) (right panel).}
\label{fig:zp}
\end{figure}

\section{Complex Cubic Ensemble of Random Matrices}

Consider the unitary ensemble of random matrices with the cubic potential\
\[
V(M)=-\frac{1}{3}M^3+tM,
\]
where $t$ is a complex parameter. The partition function of this ensemble is formally defined by the matrix integral over the space of $N\times N$ Hermitian matrices:
\[
\int_{\mathcal H_N} e^{-N\mathrm{Tr} \left(-\frac{1}{3}M^3+tM\right)} dM.
\]
Then the formal partition functions of the eigenvalues becomes
\[
\int_{-\infty}^\infty\ldots\int_{-\infty}^\infty \prod_{1\leq j<k\leq N}(s_j-s_k)^2\, \prod_{j=1}^N e^{-N V(s_j;t)} ds_1\ldots ds_N,
\]
where \( V(s;t) \) was defined above in \eqref{Ai}. This expression is formal because the integrals are divergent and need regularization. To achieve it, let
\[
\Gamma = \Gamma(\lambda) := \alpha_0 L_0 + \alpha_1 L_1 + \alpha_2 L_2,
\]
where the rays \( L_k \) were defined right after \eqref{Ai} and  \( \alpha_0,\alpha_1,\alpha_2\in\C \) are complex parameters such that
\begin{equation}
\label{alphas}
\alpha_0 = \frac\lambda\pi, \quad \alpha_1 = -\frac{\lambda}{2\pi} + \frac1{2\pi\ic}, \qandq \alpha_2 = -\frac{\lambda}{2\pi} -  \frac1{2\pi\ic},
\end{equation}
when \( |\lambda|<\infty \) and \( \alpha_0=1/\pi \), \( \alpha_1=\alpha_2 = -1/2\pi \) when \( \lambda=\infty \). Let
\[
Z_N(t) := \int_{\Gamma}\ldots\int_{\Gamma} \prod_{1\leq j<k\leq N}(s_j-s_k)^2\, \prod_{j=1}^N e^{-N V(s_j;t)}  ds_1\ldots ds_N,
\]
which is well defined (the integrals are convergent) for all values $t\in \C$, where we understand \( \int_\Gamma \) as \( \sum \alpha_k\int_{\Gamma_k}\).  As shown in \cite{MR3071662},  the topological expansion of the free energy
\[
F_N(t) = \frac1{N^2} \log Z_N(t)
\]
of the cubic ensemble of random matrices is connected to the enumeration of regular graphs of degree 3 on Riemann surfaces. The partition function \( Z_N(t) \) can also be expressed as
\begin{equation}
\label{partition}
Z_N(t) = N! D_{N-1}(t,N), \quad D_n(t,N) := \left[\int_\Gamma s^{i+j} e^{-NV(s;t)}ds\right]_{i,j=0}^n.
\end{equation}
The entries of the matrix defining \( D_n(-t,N) \) can be written as
\begin{align*}
\int_\Gamma s^k e^{-NV(s;-t)}ds &= \frac1{N^{(k+1)/3}} \int_\Gamma s^k e^{-V\left(s;-N^{2/3}t\right)}ds \\
& = \frac1{N^{k+1/3}} \frac{d^k}{dt^k} \left( \int_\Gamma e^{-V\left(s;-N^{2/3}t\right)}ds \right) \\
& = \frac1{N^{k+1/3}} \frac{d^k}{dt^k} \left( \alpha_0\pi \, \Bi\left(-N^{2/3}t\right) + (\alpha_1-\alpha_2)\pi\ic \, \Ai\left(-N^{2/3}t\right)\right),
\end{align*}
where we used expressions \eqref{Ai} and \eqref{Bi} as well as the condition \( \alpha_0+\alpha_1+\alpha_2=0 \).  That is, upon setting \( z= (\sqrt 2N)^{2/3} t \), we get that
\[
\int_\Gamma s^k e^{-NV(s;-t)}ds = \frac{2^{k/3}}{N^{(k+1)/3}} \frac{d^k}{dz^k} \left( \alpha_0\pi \, \Bi\left(-2^{-1/3}z\right) + (\alpha_1-\alpha_2)\pi\ic \, \Ai\left(-2^{-1/3}z\right)\right).
\]
This formula and the Hankel structure of the determinant \( D_n(-t,N) \) then imply that
\begin{equation}
\label{Dtau}
D_n(-t,N) = \frac{2^{n(n+1)/3}}{N^{(n+1)(n+2)/3}} \tau_{n+1}(z;\lambda ),  \quad z= (\sqrt 2N)^{2/3} t,
\end{equation}
where \( \tau_n(z;\lambda) \) was defined in \eqref{taun}. In particular, if we take \( n=N \),  the last two representations in \eqref{Airy_sol} can be rewritten in terms of the free energy \( F_N(t) \) as
\[
\begin{cases}
p_N(z;\lambda) & \displaystyle = -2N^2\frac{d^2}{dz^2}F_N\left(-(\sqrt 2N)^{-2/3} z\right), \smallskip \\
\sigma_N(z;\lambda) & \displaystyle = N^2\frac{d}{dz} F_N\left(-(\sqrt 2N)^{-2/3} z\right).
\end{cases}
\]

The cases of particular interest to us in this note are \( \lambda=0,-\ic,\ic \). These cases are essentially the same as they correspond to taking the seed function in  \eqref{taun} to be
\[
\Ai\left(-2^{-1/3}z\right), \quad \Ai\left(-2^{-1/3}z e^{2\pi\ic/3}\right),\quad \text{and} \quad \Ai\left(-2^{-1/3}z e^{-2\pi\ic/3}\right)
\]
since \( \Ai(z) \pm \ic \Bi(z) = 2e^{\pm2\pi\ic/3} \Ai(ze^{\mp2\pi\ic/3}) \), see \cite[Equation~(9.2.1)]{DLMF}. It is worth pointing out that, in these special cases, the functions \( \tau_{n+1} \) also appear in the study of the scaling limits of the largest eigenvalue in the GUE ensemble, see \cite[Proposition~28]{MR1833807}.

The above representations can be further rewritten utilizing the connection to orthogonal polynomials. Let \( P_n(s;t,N) \) be a non-identically zero polynomial of degree at most \( n \) such that
\begin{equation}
\label{ortho}
\int_{\Gamma} s^kP_n(s;t,N)e^{-NV(s;t)} ds = 0,\quad k\in\{0,\ldots, n-1\}.
\end{equation}
Due to the non-Hermitian character of the above relations, it might happen that a polynomial satisfying \eqref{ortho} is non-unique. In this case we understand by $P_n(z;t,N)$ the monic non-identically zero polynomial of the smallest degree (such polynomial is always unique). The standard determinantal representation of orthogonal polynomials yields that 
\begin{equation}
\label{hnN}
h_n(t,N) := \int_{\Gamma} P_n^2(s;t,N)e^{-NV(s;t)} ds = \frac{D_n(t,N)}{D_{n-1}(t,N)},
\end{equation}
where \( D_{-1}(t,N)\equiv 1 \). Observe that \( D_n(t,N) \) is an entire function of \( t \) and therefore each \( h_n(t,N) \) is meromorphic in \( \C \). Hence, given \( n \), the set of the values \( t \) for which there exists \( k\in\{0,\ldots,n \} \) such that \( h_k(t,N) =0 \) is countable with no limit points in the finite plane. Outside of this set the standard argument using \eqref{ortho} shows that
\begin{equation}
\label{recurrence}
sP_n(s;t,N) = P_{n+1}(s;t,N)+\beta_n(t,N) P_n(s;t,N)+\gamma_n^2(t,N) P_{n-1}(s;t,N),
\end{equation}
and by analytic continuation this relation extends to those values of \( t \) for which we have that \( D_{n-1}(t,N)D_n(t,N) \neq 0 \) (that is, \( n+1 \)-st and \( n \)-th polynomials appearing in \eqref{recurrence} have the prescribed degrees), where
\begin{equation}
\label{gammanN}
\gamma_n^2(t,N) = \frac{h_n(t,N)}{h_{n-1}(t,N)}.
\end{equation}
Finally, denote by \( p_{n,n-1}(t,N) \) the coefficient of \( P_n(s;t,N) \) next to \( s^{n-1} \), where \( p_{0,-1}(t,N) := 0 \). It easily follows from \eqref{recurrence} that
\begin{equation}
\label{betanN}
\beta_n(t,N) = p_{n,n-1}(t,N) - p_{n+1,n}(t,N).
\end{equation}
In what follows, when it is important to us to stress that the quantities appearing in \eqref{partition}, \eqref{ortho},  \eqref{hnN}, \eqref{gammanN}, and \eqref{betanN} depend on \( \lambda \), we shall use superscript \( (\lambda) \).

\begin{theorem}
\label{thm:airy}
Fix \( N\geq 1 \). Given \( \lambda\in\overline\C \), it holds for each \( n\geq 1 \) that
\begin{equation}
\label{Airy_Pol}
\begin{cases}
q_n(z;\lambda) & \displaystyle = -(N/2)^{1/3} \, \beta_{n-1}^{(\lambda)} \left(-(\sqrt 2N)^{-2/3} z,N\right), \smallskip \\
p_n(z;\lambda) & \displaystyle = -2(N/2)^{2/3} \, \gamma_n^{(\lambda)} \left(-(\sqrt 2N)^{-2/3} z,N\right)^2, \smallskip \\
\sigma_n(z;\lambda) & \displaystyle = -(N/2)^{1/3} \, p_{n,n-1}^{(\lambda)}\left(-(\sqrt 2N)^{-2/3} z,N\right).
\end{cases}
\end{equation}
\end{theorem}
\begin{proof}
Since \( \lambda,N \) are fixed, we do not indicate the dependence on them of all the quantities related to orthogonal polynomials. A combination of \eqref{Airy_sol}, \eqref{Dtau}, and \eqref{hnN} immediately yields that 
\[
q_n(z;\lambda) = -\frac d{dz} \log h_{n-1}\left(-(\sqrt 2N)^{-2/3} z\right).
\]
Since \( P_n(s;t) \) is a monic polynomial, its partial \( t \) derivative is a polynomial of degree at most \( n - 1\), which is, of course,  orthogonal to \( P_n(s;t) \). Hence, we get from \eqref{hnN}, \eqref{ortho}, and \eqref{betanN} that
\begin{align}
\frac d{dt} h_n(-t) &= N\int_{\Gamma} sP_n^2(s;-t)e^{-NV(s;-t)} ds \nonumber \\
& = N\int_{\Gamma} \big(s^{n+1} + p_{n,n-1}(-t) s^n\big)P_n(s;-t)e^{-NV(s;-t)} ds \nonumber \\
& = N\int_{\Gamma} P_{n+1}(s;-t) P_n(s;-t)e^{-NV(s;-t)} ds + N \beta_n(-t) h_n(-t) \nonumber \\
& = N \beta_n(-t) h_n(-t) ,
\label{der_hn}
\end{align}
from which the formula for \( q_n(z;\lambda) \) follows. Now, we get from \eqref{hnN} that
\[
\log D_n(t) = \sum_{k=0}^n \log h_k(t).
\]
The above expression and \eqref{der_hn}, together with \eqref{Airy_sol} and \eqref{Dtau}, establish formula for \( \sigma_n(z;\lambda) \) since
\begin{equation}
\label{der_lnD}
\frac d{dt} \log D_{n-1}(-t) =N  \sum_{k=0}^{n-1} \beta_k(-t) = -Np_{n,n-1}(-t).
\end{equation}
Finally, we get from orthogonality relations \eqref{ortho} that
\begin{align}
0 & = \frac d{dt} \left( \int_\Gamma  P_n(s;-t) P_{n-1}(s;-t)e^{-NV(s;-t)} ds \right)  \nonumber \\
&=  \int_\Gamma  \left( NsP_n(s;-t) -p_{n,n-1}^\prime(-t) s^{n-1} \right)P_{n-1}(s;-t)e^{-NV(s;-t)} ds \nonumber  \\
\label{der_pn}
& = Nh_n(-t) -p_{n,n-1}^\prime(-t) h_{n-1}(-t),
\end{align}
which, together with \eqref{Airy_sol}, \eqref{Dtau}, and  \eqref{gammanN}, gives the expression for \( p_n(z;\lambda) \).
\end{proof}

The results of the previous theorem can equivalently be rewritten in the following form.
\begin{corollary}
\label{cor:airy}
Given \( \lambda\in\overline\C \), it holds for each \( n\geq 1 \) that
\begin{equation}
\label{Airy_Pol1}
\begin{cases}
q_n(z;\lambda) & \displaystyle = -2^{-1/3} \, \beta_{n-1}^{(\lambda)} \big(-2^{-1/3} z,1\big), \smallskip \\
p_n(z;\lambda) & \displaystyle = -2^{1/3} \, \gamma_n^{(\lambda)} \big(-2^{-1/3} z,1\big)^2, \smallskip \\
\sigma_n(z;\lambda) & \displaystyle = -2^{-1/3} \, p_{n,n-1}^{(\lambda)} \big(-2^{-1/3} z,1\big).
\end{cases}
\end{equation}
\end{corollary}
\begin{proof}
Again, we do not indicate the dependence on \( \lambda \). It readily follows from \eqref{ortho} that
\[
0 = \int_{\Gamma} s^kP_n(s;t,N)e^{-NV(s;t)} ds = N^{-(k+1)/3}\int_\Gamma s^k P_n\left(N^{-1/3}s;t,N\right) e^{-V(s;N^{2/3}t)}ds.
\]
Therefore, it holds that
\[
P_n(s;t,N) = N^{-n/3} P_n\big(N^{1/3}s;N^{2/3}t,1\big)
\]
when \( \deg P_n =n \), which happens for the values of \( t \) outside of a countable set without limit points in the finite plane. From this we can readily deduce that
\begin{equation}
\label{scal_form}
\begin{cases}
\beta_{n-1}(t;N) & = N^{-1/3}\beta_{n-1}\left(N^{2/3}t,1\right), \smallskip \\
\gamma_n^2(t;N) & = N^{-2/3}\gamma_n^2\left(N^{2/3}t,1\right), \smallskip \\
p_{n,n-1}(t,N) & = N^{-1/3}p_{n,n-1}\left(N^{2/3}t,1\right).
\end{cases}
\end{equation}
The above relations and \eqref{Airy_Pol} easily yield \eqref{Airy_Pol1}.
\end{proof}

\section{Scaling Limits}

Now we can formulate a result on the asymptotic behavior of the Airy solutions when \( \lambda\in\{0,-\ic,\ic \} \). To this end, it was shown in \cite[Section~5]{MR3607591} that the critical graph of a quadratic differential
\[
 -(1+1/s)^3 ds^2,
\]
say $\mathcal C$, consists of 5 critical trajectories emanating from $-1$ at the angles $2\pi k/5$, $k\in\{0,1,2,3,4\}$, one of them being $(-1,0)$, other two forming a loop crossing the real line approximately at $0.635$, and the last two approaching infinity along the imaginary axis without changing the half-plane (upper or lower), see Figure~\ref{fig:loops}~(left panel). 
\begin{figure}[ht!]
\centering

\begin{subfigure}{.14\textwidth}
\includegraphics[width=\textwidth]{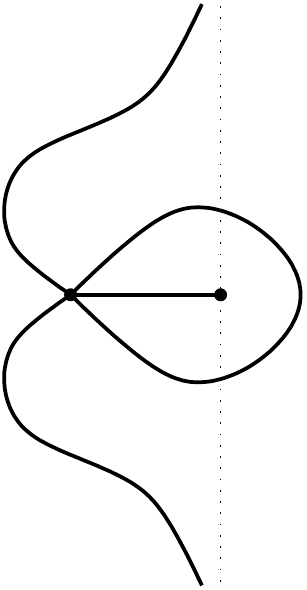}
\begin{picture}(0,0)
\put(37,53){\small $0$}
\put(-4,53){\small $-1$}
\end{picture}
\end{subfigure}
\quad\quad
\begin{subfigure}{.25\textwidth}
\includegraphics[width=\textwidth]{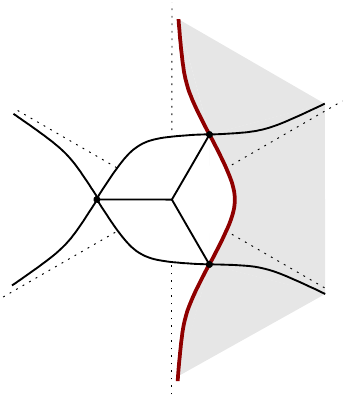}
\begin{picture}(0,0)
\put(-5,55){${\footnotesize -\sqrt[3]{1/2}}$}
\put(61,57){$\Omega_{(0)}$}
\end{picture}
\end{subfigure}
\quad\quad
\begin{subfigure}{.3\textwidth}
\includegraphics[width=\textwidth]{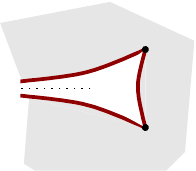}
\begin{picture}(0,0)
\put(50,22){\footnotesize $3\cdot2^{-2/3} e^{-2\pi\ic/3} $}
\put(50,75){\footnotesize $3\cdot2^{-2/3} e^{2\pi\ic/3}$}
\put(75,50){$O_{(0)}$}
\end{picture}
\end{subfigure}
\caption{\small Left panel: the critical graph $\mathcal C$; middle panel: the set $\Delta$ (solid lines) and the domain $\Omega_{(0)}$ (shaded region); right panel: the domain $O_{(0)}$ (shaded region).}
\label{fig:loops}
\end{figure}
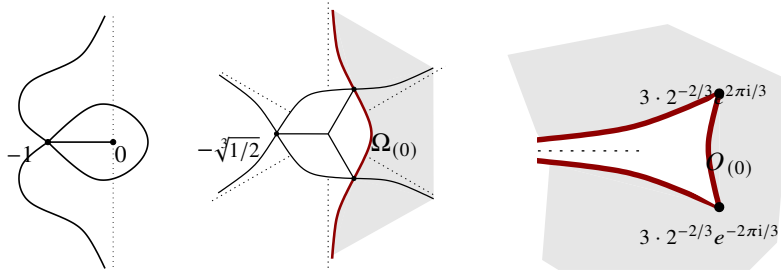
Define
\[
\Delta:=\big\{x:~2x^3\in\mathcal C\big\}
\]
and put $\Omega_{(0)}$ to be the subset of the right-half plane bounded by three smooth subarcs\footnote{These subarcs are the one originating at \( e^{\pi\ic/3}\sqrt[3]{1/2} \) and extending to infinity in the direction of the angle \( \pi/2 \), the one connecting \( e^{\pi\ic/3}\sqrt[3]{1/2} \) and \( e^{-\pi\ic/3}\sqrt[3]{1/2} \), and the one originating at \( e^{-\pi\ic/3}\sqrt[3]{1/2} \) and extending to infinity in the direction of the angle \( -\pi/2 \).} of \( \Delta \) as on Figure~\ref{fig:loops}~(middle panel). Further put \( \Omega_{(\pm\ic)} := e^{\mp2\pi\ic/3}\Omega_{(0)}\). Let
\[
t(x) := (x^3-1)/x.
\]
The function $t(x)$ is holomorphic in each $\Omega_{(\lambda)}$, \( \lambda\in\{0,-\ic,\ic\} \), with non-vanishing derivative there. Set
\[
O_{(\lambda)} := t(\Omega_{(\lambda)}),
\]
see~Figure~\ref{fig:loops}~(right panel). The inverse map $x_\lambda(t)$ exists and is holomorphic in $O_{(\lambda)}$.  One can readily check that \(  O_{(\pm\ic)} = e^{\pm2\pi\ic/3} O_{(0)} \) and
\begin{equation}
\label{rel1}
x_0(t) = e^{\pm2\pi\ic/3} x_{\pm\ic}\left( te^{\pm2\pi\ic/3}\right).
\end{equation}
Clearly, the branch \( x_0(t) \) is positive on \( O_{(0)} \cap (-\infty,\infty) \) (it was shown in \cite[Theorem~1.1]{MR3218792} that \( \partial O_{(0)} \) intersects the real line at \( t_0 \sim -1.0005424 \)) and is equal to \( 1 \) at the origin.

\begin{theorem}
\label{thm:scal_lim}
Let \( \lambda\in\{0,-\ic,\ic\} \). For all \( n \) large enough it holds that
\[
\begin{cases}
\displaystyle (2/n)^{1/3} q_n\left(-(\sqrt2n)^{2/3}t;\lambda\right) & = -x_\lambda(t) + \mathcal O\big(n^{-1}\big), \smallskip \\
\displaystyle (2/n)^{2/3} p_n\left(-(\sqrt2n)^{2/3}t;\lambda\right) & = 1/x_\lambda(t) + \mathcal O\big(n^{-2}\big), \smallskip \\
\displaystyle 2^{1/3}n^{-4/3} \sigma_n\left(-(\sqrt2n)^{2/3}t;\lambda\right) & = x_\lambda(t) - (2x_\lambda(t))^{-2} +  \mathcal O\big(n^{-2}\big),
\end{cases}
\]
uniformly on closed subset of \( \{t\in O_{(\lambda)}:|x_\lambda(t)| < \sqrt n/\delta\} \) in the first formula and of \( \{t\in O_{(\lambda)}: |\arg(t-t_\lambda)-\theta_\lambda|>\delta\} \) in the other two formulae, where \( \delta>0 \) is any, \( \theta_0=\pi \), \( \theta_{\pm\ic}=\mp\pi/3 \), and \( t_\lambda = e^{2\pi\lambda/3}t_0 \).
\end{theorem}
\begin{proof}
It holds for \( \lambda\in\{-\ic,\ic\} \) that
\[
P_n^{(0)}(s;t,N) = e^{2\pi\lambda n/3} P_n^{(\lambda)}\left( se^{-2\pi\lambda/3};te^{2\pi\lambda/3},N \right),
\]
see \eqref{ortho}. This implies that
\[
\begin{cases}
\gamma_n^{(0)}(t,N)^2 & = e^{4\pi\lambda/3} \gamma_n^{(\lambda)} \big( te^{2\pi\lambda/3},N )^2, \smallskip \\
\beta_n^{(0)}(t,N) & = e^{2\pi\lambda/3} \beta_n^{(\lambda)} \big( te^{2\pi\lambda/3},N ), \smallskip \\
p_{n,n-1}^{(0)}(t,N) & = e^{2\pi\lambda/3} p_{n,n-1}^{(\lambda)} \big( te^{2\pi\lambda/3},N ).
\end{cases}
\]
It has been shown in \cite[Theorem~4.5]{MR3607591} that for \( |n-N|\leq N_0 \) it holds that
\begin{equation}
\label{asymp_exp}
\begin{cases}
\gamma_n^{(-\ic)}(t,N)^2 & \sim -1/(2x_{-\ic}(t)) + \sum_{k=1}^\infty G_k(t;n-N) N^{-k}, \medskip \\
\beta_n^{(-\ic)}(t,N) & \sim x_{-\ic}(t) + \sum_{k=1}^\infty B_k(t;n-N) N^{-k},
\end{cases}
\end{equation}
where the expansions hold uniformly on compact subsets of \( O_{(-\ic)} \) and closed subsets of
\[
O_{(-\ic),\delta} := \{t\in O_{(-\ic)}: |\arg(t-t_{-\ic})-\theta_{-\ic}|>\delta\}
\]
in the case of \( \gamma_N^{(-\ic)}(t,N) \) (the actual condition of separation from \( \partial O_{(-\ic)} \) is more refined in \cite[Theorem~4.5]{MR3607591}, see \cite[Definition~4.1]{MR3607591}). The functions \( G_k(t;n-N) \) and \( B_k(t;n-N) \) are analytic in \( O_{(-\ic)} \) and \( G_{2k-1}(t;0) \equiv 0 \). It is also implicitly contained in \cite[Section~8.2]{MR3607591} that
\begin{equation}
\label{asymp_exp2}
p_{n,n-1}^{(-\ic)}(t,N) = -nx_{-\ic}(t) + N(2x_{-\ic}(t))^{-2} + \mathcal O\big(N^{-1}\big),
\end{equation}
where the error term is again analytic in \( O_{(-\ic)} \), uniform on compact subsets of \( O_{(-\ic)} \) and closed subsets of \( O_{(-\ic),\delta} \) when \( n=N \). The second and third claims of the theorem now follow from Theorem~\ref{thm:airy} by taking \( N=n \) in \eqref{asymp_exp2} and in the first formula of \eqref{asymp_exp}. The first claim follows by replacing \( n \) with \( n-1 \) in the second line of \eqref{asymp_exp} and taking \( N=n \) there, except this proves it only on compact subsets of \( O_{(\lambda)} \). To overcome this difficulty, write
\begin{equation}
\label{gamma_nn}
\gamma_n^2(t,n) = -1/(2x_\lambda(t)) + E_{\gamma_n}(t),
\end{equation}
where we suppress the dependence on \( \lambda \) of the recurrence coefficients as emphasizing it is no longer important. The error term \( E_{\gamma_n}(t) \) is such that
\[
|E_{\gamma_n}(t)| \leq C_F /n^2, \quad t\in F,
\]
for each closed subset \( F \subset O_{(\lambda),\delta} \) and any \( \delta>0 \). Notice that \( F_\epsilon = \{t:\dist(t,F)\leq \epsilon\} \) belongs to \( O_{(\lambda),\delta^\prime} \) for some appropriate choices of \( \epsilon,\delta^\prime>0 \). Applying Cauchy integral formula on circles of radius \( \epsilon \) around every point of \( F \) yields that \( E_{\gamma_n}^\prime(t) \) satisfies the same type of bound as \( E_{\gamma_n}(t) \), possibly with a different constant. Now, let us write
\begin{equation}
\label{beta_nn}
\beta_{n-1}(t,n) = x_\lambda(t) + E_{\beta_{n-1}}(t).
\end{equation}
It follows from Theorem~\ref{thm:airy} and the second relation in \eqref{ham_sys} that
\[
2n\beta_{n-1}(t,n)\gamma_n^2(t,n) + n = (\gamma_n^2)^\prime(t,n). 
\]
Therefore, we get from \eqref{gamma_nn} and \eqref{beta_nn} that
\[
E_{\beta_{n-1}}(t) \big( 1 - 2x_\lambda(t) E_{\gamma_n}(t) \big) = 2x_\lambda^2(t)E_{\gamma_n}(t) - \frac1{2n} \left( \frac {x^\prime_\lambda(t)}{x_\lambda(t)} + 2x_\lambda(t)E_{\gamma_n}^\prime(t) \right).
\]
Because \( x^3-xt-1=0 \), we have that \( x^\prime = x^2 / (2x^3+1) \) for any branch \( x(t) \). Since we already have the estimate of \( E_{\beta_{n-1}}(t) \) on compact subsets of \( O_{(\lambda)} \), we can assume without loss of generality that \( |x_\lambda(t)| \geq 1 \). Thus, it holds that
\[
E_{\beta_{n-1}}(t) \big( 1 - \mathcal O_\delta\big(n^{-3/2}\big) \big) = \mathcal O_\delta\big(n^{-1}\big)
\]
uniformly on closed subsets \( O_{(\lambda)} \) for which \( 1\leq|x_\lambda(t)| \leq \sqrt n/\delta \). This, of course, finishes the proof of the theorem. 
\end{proof}

Theorem~\ref{thm:scal_lim} yields the following corollary, which gives pole and zero free regions around the origin in the \( t \) plane for the Painlev\'e II functions. The results in \cite{MR3860606} in this direction are asymptotic in \( z \) without uniformity in the parameter \( n \), but in our case we can show that the radius of these domains increases with \( n \), something that is consistent with previous numerical experiments.

\begin{corollary}
For each \( t_*>0 \) such that \( \{|t|\leq t_*\} \subset O_{(\lambda)} \) there exists \( n_* \) such that the functions \( p_n(z;\lambda) \), \( q_n(z;\lambda) \) and \( \sigma_n(z;\lambda) \) are analytic and non-vanishing in \( \big\{|z|\leq (\sqrt2n)^{2/3}t_* \big\} \) for all \( n\geq n_* \).
\end{corollary}

The asymptotic behavior of the above quantities in the complement of \( \overline O_{(-\ic)} \) was studied in \cite{MR4436195}. However, the results of \cite[Theorem~4.2]{MR4436195} on asymptotic behavior of \( \beta_n \) and \( \gamma_n^2 \) are not complete, not in the least due to the observed presence of their poles in this region.

\section{Monotonicity of Airy Solutions}

It was conjectured by Clarkson in \cite{MR3529955}, see also \cite{MR3516120}, that
\begin{equation}
\label{conj1}
q_{n+1}(z;0) < q_n(z;0) 
\end{equation}
is true for every \( z\leq 0 \) and every \( n\geq 1 \). In light of Corollary~\ref{cor:airy}, this claim is equivalent to \( \beta_n(t,1) > \beta_{n-1}(t,1) \) being true for \( t \geq 0 \) and \( n\geq 1 \). Since \cite{MR3516120} uses slightly different notation, let us point out that
\[
D_{n-1}(t,1) =  2^{n(n-1)/3}\det \left[ \frac{d^{j+k}}{dt^{j+k}} \Ai(t)  \right]_{j,k=0}^{n-1},
\]
where everything related to orthogonal polynomials corresponds to \( \lambda=0 \) in \eqref{alphas}, see \eqref{Dtau} and \eqref{taun} (\( C_1=1 \) and \( C_2=0 \)). This quantity was labeled by \( \Delta_n(t) \) in \cite[Equation~(5.2)]{MR3516120}. Hence, it follows from \eqref{der_lnD}, \eqref{der_pn}, and \eqref{gammanN} as well as \eqref{der_hn} and \eqref{hnN} that
\begin{equation}
\label{gb}
\gamma_n^2(t,1) = \frac{d^2}{dt^2}\log D_{n-1}(t,1) \qandq \beta_n(t,1) = \frac{d}{dt} \log \frac{D_{n-1}(t,1)}{D_n(t,1)}.
\end{equation}
These quantities were labeled as \( -a_n(t) \) and \( b_{n-1}(t) \) in \cite[Equation~(5.1)]{MR3516120}. In \cite{MR3516120}, these functions were studied since  they are solutions of an \emph{alternative discrete Painlev\'e I system}:
\[
\begin{cases}
\gamma_n^2(t,1) + \gamma_{n+1}^2(t,1) + \beta_{n+1}^2(t,1) = t \smallskip \\
\gamma_n^2(t,1) \big( \beta_{n+1}(t,1) + \beta_n(t,1) \big) = -n,
\end{cases}
\]
see \cite{MR1251864}. Using rescaling formulae \eqref{scal_form},  the above relations can be rewritten as
\[
\begin{cases}
\gamma_n^2(t,N) + \gamma_{n+1}^2(t,N) + \beta_{n+1}^2(t,N) = t \smallskip \\
\gamma_n^2(t,N) \big( \beta_{n+1}(t,N) + \beta_n(t,N) \big) = -n/N,
\end{cases}
\]
which are also known as \emph{discrete string equations}, see \cite{MR3155177}. It was conjectured in \cite[Conjectures~3.2 and~5.3]{MR3516120} that
\begin{equation}
\label{conj3}
0 < \beta_n^\prime(t,1) < 1/(2\sqrt t) \qandq \sqrt t < \beta_n(t,1) < \beta_{n+1}(t,1)
\end{equation}
for all \( t>0 \) and \( n\geq 0 \). It was also shown that the first inequality in \eqref{conj3} implies the last inequality there and therefore \eqref{conj1}. To see this, observe that it follows from the difference of the successive top string equations, \eqref{gb}, and the bottom string equation that
\[
\beta_{n+1}(t,1) - \beta_n(t,1) = \frac{\gamma_{n-1}^2(t,1) - \gamma_{n+1}^2(t,1)}{\beta_{n+1}(t,1) + \beta_n(t,1)} = -(\beta_n^\prime(t,1) + \beta_{n-1}^\prime(t,1))\frac{\gamma_n^2(t,1)}n.
\]
Now, it was shown in \cite[Equation~(4.17)]{MR4031470} that
\[
D_{n+1}(t,1) = -(n+1)D_{n-1}(t,1)\int_t^\infty \frac{D_n^2(s,1)}{D_{n-1}^2(s,1)}ds
\]
for \( n\geq 0 \), where \( D_0(t,1) = \Ai(t) \) and \( D_{-1}(t,1) \equiv 1 \). Let \( \iota_1 <0 \) be the largest real zero of \( \Ai(t) \). Then it follows from the above relation and \eqref{hnN} that \( \sgn \, h_n(t,1) = (-1)^n \) for \( t\in(\iota_1,\infty) \). This, of course, means that \( \gamma_n^2(t,1)<0 \) there by \eqref{gammanN}, which yields the desired implication on \( (\iota_1,\infty) \).

\begin{theorem}
\label{thm:clarkson}
There exists \( n_0>0 \) such that for all \( n\geq n_0 \) it holds that 
\[
\beta_{n-1}^\prime(t,1)>0, \quad t\in[\iota_1,\infty).
\]
In particular, \eqref{conj1} holds for \( z \leq -2^{1/3} \iota_1 \) and all \( n\geq n_0\).
\end{theorem}
\begin{proof}
It was shown in  \cite[Lemma~3.3]{MR3516120} that \( \beta_n^\prime(t,1)>0 \) for \( t\geq 2^{-2/3}n^{4/3} \). Thus, we are only interested in \( t\in [\iota_1,2^{-2/3}n^{4/3}] \). It follows from \eqref{scal_form} that
\[
\beta_{n-1}(t,1) = n^{1/3}\beta_{n-1}(n^{-2/3}t,n).
\]
Hence, it is sufficient for us to show that \( \beta_n^\prime(t,n) >0 \) for \( t\in F_n:=[-\epsilon,n^{2/3} ] \), where \( \epsilon>0 \) and small. To this end, recall that \( x^\prime = x^2 / (2x^3+1) \). Observe also that \( x_0(t)>0 \) for \( t\geq-\epsilon \) (this is the branch that is equal to \( 1 \)  at \( 0 \) and is positive for \( t\geq 0 \)). Hence, \( x_0^\prime(t)>0 \) for \( t\geq-\epsilon \) and therefore \( x_0(t) \) is increasing there. Notice that \( x_0(t) < 2\sqrt t \) for \( t\geq1 \) (constant \( 2 \) is by no means optimal). Indeed, otherwise we would have that
\[
1 = x_0(t) \big( x_0^2(t) - t \big) > 6 t^{3/2} \geq 6,  
\]
which is impossible. Thus, \( F_n \) lies within the set where \( |x_0(t)| \leq 2n^{1/3} < 2n^{1/2} \). Moreover, there exists an open set \( U_n \supset F_n \) satisfying \( \dist(\partial U_n,F_n)\geq c>0 \) that also lies within \( \{t:|x_0(t)| \leq 2n^{1/2}\} \). As in the proof of Theorem~\ref{thm:scal_lim}, we can write
\[
\beta_{n-1}(t,n)  = x_0(t) + E_{\beta_{n-1}}(t), 
\]
where \( E_{\beta_{n-1}}(t) \) is analytic in \( O_{(0)} \) and satisfies (through the use of Cauchy integral formula)
\[
|E_{\beta_{n-1}}^\prime(t)| \leq K/n, \quad t\in F_n,
\]
for some constant \( K>0 \). Since we can take \( \epsilon \) small enough so that \( 2x_0^3(-\epsilon)>1 \), it then holds that
\[
\beta_{n-1}^\prime(t,n)  = x_0^\prime(t) + E_{\beta_{n-1}}^\prime(t) > \frac1{4x_0(t)} - \frac Kn \geq \frac1{4x_0(n^{2/3})} - \frac Kn > \frac{n^{2/3}-8K}{8n}
\]
for \( t\in F_n \), which finishes the proof of the theorem.
\end{proof}

\small

\bibliographystyle{plain}

\bibliography{airy}

\end{document}